\documentclass[11pt,a4paper]{article}
\usepackage[T1]{fontenc}
\usepackage[utf8]{inputenc}
\usepackage{lmodern,microtype}
\usepackage[margin=25mm,headheight=14pt]{geometry}
\usepackage{amsmath,amssymb,amsthm,mathtools}
\usepackage{booktabs,enumitem}
\usepackage{xcolor,fancyhdr,graphicx}
\usepackage[hidelinks,breaklinks]{hyperref}
\usepackage{url}
\setlist{itemsep=2pt,topsep=4pt,leftmargin=*}
\newtheorem{theorem}{Theorem}[section]
\newtheorem{proposition}[theorem]{Proposition}
\newtheorem{lemma}[theorem]{Lemma}
\newtheorem{corollary}[theorem]{Corollary}
\theoremstyle{definition}
\newtheorem{definition}[theorem]{Definition}
\theoremstyle{remark}

\newcommand{\Encode}{\operatorname{Encode}}
\newcommand{\View}{\operatorname{View}}
\newcommand{\Roots}{\operatorname{Roots}}
\newcommand{\Sign}{\operatorname{Sign}}
\newcommand{\Verify}{\operatorname{Verify}}
\newcommand{\Adv}{\operatorname{Adv}}

\newcommand{\argmax}{\operatorname*{arg\,max}}
\newcommand{\botout}{\mathsf{unresolved}}

\title{\vspace{-12mm}\LARGE\bfseries Provenance and Authorization in\protect\\Autobiographical Language Models\protect\\[4pt]\large A Security Model and Synthetic Memory Study}
\author{%
Faruk Alpay$^{1}$\thanks{Correspondence: \href{mailto:alpay@lightcap.ai}{alpay@lightcap.ai}}\qquad Taylan Alpay$^{2}$\\[6pt]
\small $^{1}$Department of Computer Engineering, Bahçeşehir University,\\[-2pt]
\small Istanbul, Türkiye\\[3pt]
\small $^{2}$Department of Aerospace, University of Turkish Aeronautical Association,\\[-2pt]
\small Ankara, Türkiye\\[5pt]
\small \href{mailto:faruk.alpay@bahcesehir.edu.tr}{faruk.alpay@bahcesehir.edu.tr}\\[-2pt]
\small \href{mailto:s220112602@stu.thk.edu.tr}{s220112602@stu.thk.edu.tr}}
\date{12 October 2026}
\hypersetup{pdftitle={Provenance and Authorization in Autobiographical Language Models},pdfauthor={Faruk Alpay; Taylan Alpay},pdfsubject={Theoretical computer security; authenticated memory; decision integrity}}
\begin{document}
\maketitle
\thispagestyle{empty}
\begin{abstract}
Autobiographical language models may draw on archives containing observations, quotations, emotional reports, and earlier model outputs. Cryptographic authentication can preserve an archive while leaving these sources semantically conflated. We formulate this gap as a computer-security problem: narrative content must not silently become consent or alter an authorized decision rule. We define typed records, authenticated attestations, a separate policy channel, and admissible views. Under explicit assumptions, we give a standard integrity reduction and a noninterference argument, characterize the information required by exact decision-preserving summaries, and show why passive history and fixed-point convergence need not identify a preferred action. Alpay Algebra supplies a comparative language for iterative memory without being a premise of the proofs. Sociology of religious affiliation, source-monitoring psychology, and self-determination theory motivate distinctions in a bounded Latter-day Saint vignette. A separate, reproducible pilot uses a frozen public language model on 32 wholly synthetic cases across 576 paired conditions. It compares flattened records, source labels, and an explicit admission gate under contradictory summary repetition. The study concerns computational attribution and authorization, not the motives or psychology of the people in the vignette.
\end{abstract}
\noindent\textbf{Intended subject area:} Cryptography and Security (cs.CR); secondary relevance to artificial intelligence and information retrieval.\\
\textbf{Keywords:} autobiographical memory; provenance; authorization; noninterference; memory poisoning; decision sufficiency; religious affiliation.

\section{Introduction}
A personal language model can receive the same episode in several forms: a contemporaneous diary, a message attributed to another person, the diarist's interpretation, an assistant's elaboration, and a later summary. These are not interchangeable assertions. A signature on the archive can establish that particular bytes were attested, yet cannot establish that every interpretation is true, that a quoted speaker signed the text, or that an expression of warmth authorizes a consequential action.

This problem is especially visible when an archive is used to discuss a long-term personal decision. In our motivating vignette, a diarist explores affiliation with the Church of Jesus Christ of Latter-day Saints (LDS). The supplied narrative separates friendship, religious learning, and attendance from a decision about baptism. Some passages also incorporate strong model-generated interpretations of other people's intentions. An archival pipeline that discards this distinction could convert commentary into purported autobiographical evidence. Nothing in this example establishes adversarial conduct by the religious participants; the failure mechanism is a property of the computational pipeline.

We investigate three questions. First, what can authentication guarantee about a mixed-source personal archive? Second, what information must survive compression if a later decision procedure is to preserve its authorized behavior? Third, does increased memory or convergence of an iterative model determine which personal action should be chosen? These questions concern integrity, authorization, and information flow rather than theological truth or conversion advice.

Our contributions are a typed archival and policy model; a conditional integrity reduction and noninterference theorem; an exact characterization of decision-sufficient summaries; and two negative results showing that neither passive observation nor stable iteration alone determines the appropriate action. A small synthetic pilot examines the difference between presenting source labels and enforcing admission before model inference. The mathematical results use standard cryptographic assumptions and elementary arguments. We claim a focused synthesis for autobiographical decision systems, not a new cryptographic primitive or priority over general results in provenance, noninterference, or preference identifiability.

\paragraph{Scope.}
The main results are theoretical definitions, conditional proofs, and constructed counterexamples. Section~\ref{sec:pilot} adds a finite synthetic experiment with a frozen model; no model is trained on the personal narrative. The vignette is a redacted illustration from supplied material, not a recruited participant study, independently verified transcript collection, or evaluation dataset. Synthetic model behavior supplies no empirical validation of a sociological or psychological claim about the diarist or community.

\section{Related Work and Conceptual Boundaries}
\paragraph{Authentication and provenance.}
We use the standard distinction between unforgeability and semantic truth in digital signatures, as developed in modern cryptography \cite{bonehshoup}. Merkle inclusion and consistency proofs, exemplified by Certificate Transparency \cite{rfc9162}, concern membership and append-only consistency relative to checkpoints. They do not prove completeness of an account of the world. W3C PROV-DM distinguishes entities, activities, and agents and records attribution and derivation \cite{prov}. Our record types add application-specific epistemic and authorization roles to such provenance relations.

\paragraph{Information flow and language-model security.}
Goguen and Meseguer separate a security policy from the system implementing it and formalize interference restrictions \cite{goguen}. CaMeL separates control and data flows and uses capabilities to constrain tool use in language-model agents \cite{camel}. AgentPoison demonstrates that agent memory or retrieval stores can be poisoned without additional model training \cite{agentpoison}. Our target is narrower than general resistance to persuasion and different from malicious tool invocation: it is the unauthorized promotion of narrative material into decision evidence or policy. The reference construction borrows the established idea of an enforcement boundary; it does not assume that labeling text in a prompt enforces that boundary.

\paragraph{Memory and learned preferences.}
MemGPT addresses limited model context through hierarchical memory management \cite{memgpt}. Such mechanisms motivate a distinction between retained records and the active representation used for inference. Training is a separate channel: extraction results show that model parameters can expose information from training data \cite{carlini}. Armstrong and Mindermann establish that behavior does not, without additional assumptions, uniquely identify a reward function and planning process \cite{armstrong}. Our two-world proposition is a deliberately restricted observation-level version of this obstacle, not a replacement for their more general analysis.

\paragraph{Agentic memory and evaluation.}
Generative Agents combines observation, reflection, retrieval, and planning to produce believable simulated behavior \cite{generativeagents}. Its architecture motivates tracking how reflections re-enter memory; believable simulation is not a validation of a particular person's psychology. LongMemEval separately tests extraction, cross-session reasoning, temporal reasoning, updates, and abstention \cite{longmemeval}. Our pilot borrows the concern with temporal and source-sensitive queries but does not run or reproduce that benchmark. AgentLAB evaluates multi-turn attacks including memory poisoning and objective drift \cite{agentlab}, underscoring why a security boundary must cover state transitions rather than one prompt. Agent Zero Memory proposes multiple organized stores and citation-constrained retrieval \cite{agentzero}. We distinguish that kind of traceability from permission: a retrieved, correctly cited sentence can still be inadmissible as evidence of consent. These works preclude any claim that structured or provenance-aware agent memory is itself new here.

\paragraph{Religious affiliation and narrative interpretation.}
Decoo and Longman examine structural ambivalence in interviews with thirteen Mormon women in Flanders \cite{decoo}. Their study provides a scholarly reason to retain distinctions among interpersonal belonging, institutional norms, and individual responses. It does not supply estimates for the diarist or for another national setting. Scharp and Beck analyze identity construction in 150 online Mormon exit narratives \cite{scharp}. That study concerns narrated identity and cannot be used to predict this person's affiliation decision or attribute motives to missionaries. We use these works to justify careful treatment of context and narration, not to classify religious participation as a security threat.

\paragraph{Psychological constructs and their computational limits.}
The source-monitoring framework studies how people attribute remembered information to its origins, including distinctions among perception, inference, and imagination \cite{johnson}. It motivates our interest in attribution, but a record type is an engineered data field, not a model of that psychological process. Self-determination theory distinguishes autonomy, competence, and relatedness; autonomy need not mean isolation from a community \cite{ryandeci}. Accordingly, enjoying a meal and preserving freedom over affiliation need not be inconsistent. Neither theory supplies a measured utility function, a diagnosis, or a mathematical reading of another person's intentions in this case.

\paragraph{Alpay Algebra.}
Alpay's first paper introduces states, adjustments, an update rule, and an evaluation map \cite{alpay1}; the second studies identity through categorical fixed points \cite{alpay2}. The third adds observer and memory constructions \cite{alpay3}. Parts IV and V discuss semantic alignment and nested update processes \cite{alpay4,alpay5}. Section~\ref{sec:fixed} states exactly which expressions we reuse. Our proofs stand independently of these papers. In particular, we do not identify an individual's identity with a mathematical fixed point or assume that a converged representation is true.

\section{Motivating Vignette and Record Semantics}\label{sec:case}
\subsection{Bounded source material}
The source material combines first-person notes, supplied message excerpts, screenshots, and quoted assistant responses. The diarist is denoted $U$; two missionaries are $D$ and $H$; two other community participants are $C$ and $P$; and a local religious leader is $L$. Identifiers are role aliases. We omit locations, exact travel schedules, photographs, and the substance of a prospective change to civil identity records. This treatment reduces disclosure but is not a guarantee against reidentification by people familiar with the events.

The narrative describes a meal with $D$ and $H$, where mutual smiles were valued despite unresolved differences. Later, $H$ asked about the remainder of the previous Sunday; $U$ answered with a photograph from a subsequent hackathon and described a project change. $U$ experienced the photograph as an expression of continuing agency. A quoted assistant passage attributed a strategic constraining intention to $H$. The visible conversational content does not establish that intention. The chronology itself matters: combining the episodes into one continuous post-church event would alter the supplied temporal relation.

Other notes report that $C$ knew about a travel plan, that an archaeological location appeared among map favorites, and that $P$ later sat nearby. These records support descriptions of reported information and seating, not conclusions about organized monitoring or private mental states. Similarly, descriptions of language differences or apparent distance do not identify a cause. A separate hackathon episode distinguishes another person's encouraging words from the diarist's feeling of being diminished by them. An everyday clothing-stain episode includes an explicitly uncertain interpretation of a question about playing guitar. Such uncertainty is part of the record and should not disappear during summarization.

The exchange with $L$ combines an acknowledgment of general risk, strong personal assurances, and a statement that the diarist should not choose baptism merely to please the community. The diarist separately asks to discuss sensitive personal matters privately and expresses a wish to proceed without haste. This is a useful case for keeping a third party's assurance, a reported personal preference, and an actual authorization event distinct. It is not an assessment of the truth of the assurances or of local law.

\subsection{Typed statements}
Let $\mathcal{P}$ be a fixed set of principal identifiers and let $\mathcal{T}$ be a finite set of source types:
\begin{equation}
\mathcal{T}=\{\mathsf{observation},\mathsf{quotation},\mathsf{feeling},
\mathsf{interpretation},\mathsf{generated},\mathsf{declaration}\}.
\end{equation}
An observation is an attested report of observation, not an automatically verified fact. A quotation preserves an attribution, which may itself be only the diarist's attestation. A feeling records the subject's self-report. An interpretation states an inference. Generated content names the generating system where known. A declaration records an expressed preference; it does not become a machine-enforceable policy merely because it uses imperative language.

An archival record has the form
\begin{equation}\label{eq:record}
r_i=(i,p_i,a_i,\tau_i,t_i^{\rm event},t_i^{\rm record},b_i,\mathsf{parents}_i,\ell_i^{\rm scope}),
\end{equation}
where $p_i$ is the attester, $a_i$ the attributed speaker, $\tau_i\in\mathcal{T}$, $b_i$ the payload, and the final field the permitted disclosure scope. Event time can be unknown. The distinction between event time and recording time prevents an evening recollection from being mistaken for a new occurrence. Scope is an access-control label, not a digital signature. An independently supplied signature from an attributed speaker can be included as evidence, but is not presumed to exist for imported messages.

These types preserve distinctions within the vignette. The meal with $D$ and $H$ records $U$'s report of warmth and mutual smiles; it does not establish consent to membership. The Sunday question and subsequent hackathon reply preserve message attribution, distinct event times, and $U$'s stated communicative intention. They do not establish that $H$ intended to confine $U$ to a social system. Similarly, the report that $C$ knew a travel plan leaves the source of that knowledge unspecified and does not establish a surveillance mechanism. The report that $P$ sat nearby can coexist with $U$'s feeling of reduced distance without establishing that $P$ recognized or yielded to a strategy.

The same distinction applies to assurances and disclosure. Statements attributed to $L$ remain quotations rather than verified guarantees of security or confidentiality. A request for private discussion records a preference and may support a disclosure restriction if separately authorized by the owner. It does not supply blanket permission to expose the sensitive topic. These are boundaries on inference from the records, not findings about the participants.

\subsection{Derivation does not create an independent source}
Records form a directed acyclic graph through parent identifiers. For a record $r$, define its primitive lineage recursively:
\begin{equation}\label{eq:roots}
\Roots(r)=\begin{cases}
\{i\},&\mathsf{parents}(r)=\varnothing,\\
\displaystyle\bigcup_{j\in\mathsf{parents}(r)}\Roots(r_j),&\text{otherwise}.
\end{cases}
\end{equation}
Copying a claim through five summaries therefore need not create five primitive sources. This is an accounting invariant, not a probabilistic independence test: distinct roots can remain correlated, and absent parent links can conceal dependence. Parent completeness is a trust assumption whenever this lineage is used for enforcement. Attesting that an assistant produced a sentence may be new evidence about the assistant's output, but is not new independent evidence for the sentence's account of another person's intentions.

\section{Threat Model and Authenticated Authorization}\label{sec:threat}
\subsection{Principals, adversary, and trusted components}
The owner $U$ controls a policy channel. An archival attester $A$ labels and records incoming material. A memory service stores records and derived summaries. A decision engine reads only an admitted view. The model, retrieval service, external writers, and generated summaries can be faulty or adversarial. In particular, an adversary may submit arbitrary narrative payloads, replay old material, remove attribution in a proposed summary, or suggest that a quotation should be treated as an instruction. This adversary is an abstract computational actor; it is not identified with any religious participant in Section~\ref{sec:case}.

We assume uncompromised owner and attester keys, an injective canonical serialization, correct domain separation, a trusted type/authorization gate, and a verifier that retains an authentic checkpoint. Type correctness at capture and completeness of declared derivations are assumptions, not consequences of a signature. An owner may sign a mistaken judgment, and an honest observer may be wrong. Our model does not prevent either. Denial of service, endpoint compromise, coerced approval, psychological manipulation of the owner, and discovery of unrecorded events are outside the proved guarantees.

\subsection{Record authentication and checkpointing}
Let $H$ be collision resistant, and let $(\mathsf{KeyGen},\Sign,\Verify)$ be an existentially unforgeable signature scheme under chosen-message attack (EUF-CMA) \cite{bonehshoup}. Primitives and adversarial advantages are parameterized by a security parameter $\kappa$, suppressed below. For fresh random salt $s_i$, define
\begin{align}
c_i&=H(\texttt{REC}\parallel\Encode(s_i,r_i)),\\
\alpha_i&=\Sign_{sk_A}(\texttt{ATTEST}\parallel c_i).
\end{align}
All fields in~\eqref{eq:record}, including type, scope, and parents, are covered. The signature says that $A$ attested this structured record. It does not say that $a_i$ authored or signed the quoted message. Salting is not used as a substitute for an encryption security proof.

Commitments and attestations are appended to a domain-separated Merkle log. A checkpoint authenticates the log identifier, size, and root; consistency proofs connect later checkpoints to a retained one \cite{rfc9162}. We use this construction only for relative append-only integrity. It supplies no guarantee that every real event was recorded, that a server serves the newest checkpoint, or that all clients see the same log without an additional anti-equivocation mechanism. A former consent record may remain in the archive while being inactive under a later authorized policy.

\subsection{Policy is a separate authenticated object}
A policy version is
\begin{equation}
P_v=(v,H(P_{v-1}),\mathsf{purposes},\mathsf{rules},\mathsf{disclosure},\mathsf{revision}),
\end{equation}
with owner signature
\begin{equation}
\beta_v=\Sign_{sk_U}(\texttt{POLICY}\parallel\Encode(P_v)).
\end{equation}
The verifier checks the owner key, domain tag, successor relation, and the locally accepted version. It requires an explicit authorization-channel event to advance the active policy. A sentence inside a diary such as ``I should agree'' cannot syntactically substitute for this event. Conversely, the owner can deliberately adopt a new preference; the model must not freeze an earlier preference merely because it was once authenticated.

Let $V_P(D,q)=\View(P,D,q)$ denote a deterministic, purpose-specific admitted view of archive $D$ for query $q$. The trusted gate verifies records, resolves authorized corrections, retains attribution and uncertainty labels, and applies $P$. The view can include a quotation as a quotation without admitting its content as a verified world fact. Records used only for expressive reflection need not be admitted to a decision-critical interface. Our proofs assume a formal record vocabulary; correct extraction of that vocabulary from arbitrary natural language is not proved.

\begin{definition}[Unauthorized promotion]
An unauthorized promotion occurs when an attested record is accepted with changed security-relevant fields, or when archival content changes the active policy without an owner-authorized policy transition. Mistaken meaning within a correctly labeled payload is a separate semantic failure.
\end{definition}

\begin{theorem}[Conditional record and policy integrity]\label{thm:auth}
For the reference verifier above, any probabilistic polynomial-time adversary that causes unauthorized promotion by changing an authenticated record or forging an owner-policy transition yields either an EUF-CMA forgery or a collision in $H$. For suitable reductions $B_1,B_2,B_3$,
\begin{equation}\label{eq:bound}
\Pr[\mathsf{Promote}]\le
\Adv^{\rm euf}_{\Sigma,A}(B_1)+
\Adv^{\rm euf}_{\Sigma,U}(B_2)+
\Adv^{\rm cr}_{H}(B_3).
\end{equation}
This statement is conditional on correct capture labels and correct verifier behavior.
\end{theorem}
\begin{proof}
Consider the first accepted violating object. If it is a record under an attester signature on a commitment never signed by the attester, it is a signature forgery. Otherwise it reuses a signed commitment while opening it to a different serialized record or salt. The two distinct hash inputs with equal digest constitute a collision. If a dishonest inclusion path or checkpoint opening is essential, the first diverging node with a matching digest similarly gives a collision, unless the checkpoint signature itself is forged; an archive checkpoint uses the attester key and is included in the first forgery case. If the violating object is a policy, a new accepted policy message not signed by the owner gives a forgery. Reusing a signed policy does not change its covered fields; stale versions and wrong parents are rejected by the verifier. Canonical encoding and distinct domain tags prevent parsing an attestation as an owner-policy authorization. Taking the union of the bad events proves~\eqref{eq:bound}. No case in the reduction proves that a signed payload is true.
\end{proof}

\begin{proposition}[Authenticity does not imply truth]
No verifier whose only semantic evidence is a correctly signed statement $p$ can, for all possible worlds, decide whether $p$ is true.
\end{proposition}
\begin{proof}
Choose two worlds with identical keys, signed bytes, and verifier inputs, but with $p$ true in one and false in the other. Allow the signer to be mistaken. The verifier has the same output distribution in both worlds and therefore cannot be correct with certainty in both.
\end{proof}

\subsection{Confidentiality is an independent requirement}
An archive can be perfectly authentic and publicly exposed. If it is stored with an untrusted service, confidentiality requires a separate construction, such as authenticated encryption under a secret key with appropriate nonce handling \cite{rfc5116}. The visibility of metadata, record lengths, and access patterns must be specified separately. A bare hash of a password or diary is neither encryption nor an access-control mechanism. Permanent public storage is therefore not a security premise of this model. Encryption of source records also does not establish privacy of a model trained on their plaintext; the extraction literature makes that distinction necessary \cite{carlini}. We do not provide an end-to-end confidentiality theorem for a trained model.

\section{Noninterference and Decision Integrity}\label{sec:ni}
\subsection{A constrained decision interface}
Fix an authorized policy $P$, a query $q$, a decision model family $\mathcal{M}$, and a finite action set $\mathcal{A}_q$. A model $m\in\mathcal{M}$ specifies admissible outcomes and a utility or preference ordering; this specification is an input, not something signatures discover. Let
\begin{equation}
\mathcal{M}_P(V)=\{m\in\mathcal{M}:m\text{ satisfies the formal constraints admitted in }V\text{ and }P\}.
\end{equation}
Write $A_m(V,q)=\argmax_{a'\in\mathcal{A}_q}u_m(a'\mid V,q)$ for the set of optimal actions under $m$. We define an ideal protected answer
\begin{equation}\label{eq:robust}
R_P(D,q)=\begin{cases}
a,&\mathcal{M}_P(V)\ne\varnothing\ \text{and}\ \forall m\in\mathcal{M}_P(V),\ A_m(V,q)=\{a\},\\
\botout,&\text{otherwise},
\end{cases}
\end{equation}
where $V=V_P(D,q)$. Inconsistent constraints yield $\botout$, avoiding a conclusion obtained by vacuous universal quantification. Ties also yield $\botout$. This symbol reports a failure of identification; it is not the personal action ``postpone.'' The action set itself can contain postponement, which may or may not be preferred.

Equation~\eqref{eq:robust} specifies an ideal, possibly computationally intractable decision rule. It is not presented as an implemented solver. A deployed approximation would require its own correctness or error guarantee. The rules defining admissible models can be incomplete or mistaken, so agreement within $\mathcal{M}_P(V)$ is a conditional result, not universal moral or religious authority.

The protected interface includes this formal answer and typed supporting statements. A trusted renderer can display those statements using fixed templates. Cryptographic witnesses may accompany them as audit metadata. We do not claim that an unconstrained model-generated essay is noninterfering merely because its last action label is checked. Choosing the candidate actions, selecting premises, resolving conflicts, and generating decision-critical explanations all belong inside the specified boundary.

\begin{definition}[View equivalence]
Archives $D,D'$ are equivalent for $(P,q)$, written $D\equiv_{P,q}D'$, if $V_P(D,q)=V_P(D',q)$. This equality includes admitted payloads, source types, attribution, uncertainty, and authorized preferences. It need not include raw archive roots or proofs about excluded records.
\end{definition}

\begin{theorem}[Noninterference of excluded content]\label{thm:ni}
Suppose the protected output has the form
\begin{equation}
O=F(P,q,V_P(D,q);\omega),
\end{equation}
where $\omega$ is randomness independent of excluded content, and $F$ has no other archive-dependent state. If $D\equiv_{P,q}D'$, the protected output distributions are equal. For a stateful system, the same conclusion holds for the full protected trace if the initial protected states, policy sequences, queries, and admitted-view sequences coincide and every state transition factors through these inputs.
\end{theorem}
\begin{proof}
For one step, couple the executions using the same random coins. The arguments to $F$ are identical, so outputs are identical under the coupling. Equality of distributions follows. For several steps, assume coupled protected states are equal before a transition. Equal policies, queries, views, and fresh coupled coins give equal next states and outputs. Induction establishes equality of the protected traces.
\end{proof}

The theorem is a standard factorization argument specialized to our view. It establishes a sufficient architecture, not an automatic property of a language model. A retriever that lets excluded text change which admitted records are selected violates the premise. So does a shared neural encoder whose parameters have already absorbed the excluded text. Checking only the final prompt would miss these channels. Timing, denial of service, and raw audit metadata are not protected observations in this statement.

\begin{lemma}[Lineage-copy invariance]
If a projection depends only on a fixed set of admissible primitive records and their authorized annotations, adding derived copies with the same primitive lineage and no new admissible annotations leaves the projection unchanged.
\end{lemma}
\begin{proof}
By~\eqref{eq:roots}, each new copy contributes a subset of roots already present. Set union is idempotent. By hypothesis the projection uses neither the number of derived copies nor their unadmitted wording; its input is unchanged.
\end{proof}

This lemma does not detect every paraphrase or forged lineage. It states the benefit of a lineage discipline under its assumptions. In the vignette, repeated assistant descriptions of a strategic conflict do not become independent evidence just by being copied into the diary. The owner's explicit later adoption of an interpretation can be recorded as a new owner judgment, with its provenance intact, rather than silently relabeled as observation.

\section{Compression and the Limits of Personal Inference}\label{sec:compression}
\subsection{Exact decision sufficiency}
Let $\mathcal{H}$ be a nonempty finite set of possible admitted histories, $\mathcal{Q}$ a family of queries, and $P$ a fixed policy. Let $f_P(h,q)$ be the protected answer that an uncompressed reference procedure would return. The answer may include typed premises and $\botout$; its precise contents must be fixed before choosing a compressor.

\begin{definition}[Decision-sufficient representation]
A compressor $C:\mathcal{H}\to\{0,1\}^m$ is exactly decision-sufficient for $(P,\mathcal{Q})$ if a decoder $G$ exists such that
\begin{equation}\label{eq:suff}
G(P,C(h),q)=f_P(h,q)\qquad\forall h\in\mathcal{H},\ q\in\mathcal{Q}.
\end{equation}
The decoder has no history-dependent side channel or external record access.
\end{definition}

Define the task-relative equivalence relation
\begin{equation}\label{eq:equiv}
h\sim h'\quad\Longleftrightarrow\quad
\forall q\in\mathcal{Q},\ f_P(h,q)=f_P(h',q).
\end{equation}
Let $N=|\mathcal{H}/{\sim}|$ be the number of equivalence classes.

\begin{theorem}[Minimum size of an exact decision summary]\label{thm:compression}
An $m$-bit exactly decision-sufficient representation exists if and only if
\begin{equation}\label{eq:bits}
m\ge\lceil\log_2 N\rceil.
\end{equation}
This is an information-theoretic characterization; it makes no efficiency claim.
\end{theorem}
\begin{proof}
If histories in two different classes receive the same representation, some query distinguishes their reference answers. The decoder receives identical inputs for that query and cannot return both answers. Thus $C$ needs at least $N$ values and $2^m\ge N$. Conversely, assign each equivalence class a distinct $m$-bit code. For each code and query, define the decoder to return the common answer for that class. Equation~\eqref{eq:equiv} makes the decoder well-defined, proving sufficiency.
\end{proof}

\begin{corollary}[No fixed-size universal exact autobiography]
If histories contain $n$ independent binary facts and the permitted query family includes recovery of each fact, then any exact standalone representation requires at least $n$ bits. No fixed finite-precision representation answers all such queries exactly for unbounded histories.
\end{corollary}
\begin{proof}
Take $\mathcal{H}=\{0,1\}^n$ and queries asking for coordinate $j$. Any two distinct histories differ at a coordinate, so every equivalence class is a singleton and $N=2^n$. Apply Theorem~\ref{thm:compression}.
\end{proof}

This identifies a precise storage limitation, distinct from the processor--memory performance gap traditionally called the memory wall \cite{memorywall}. A model with a fixed number of finitely represented weights cannot losslessly encode arbitrary amounts of queryable history. External retrieval changes the premise by supplying additional history-dependent information. A short cryptographic commitment does not contradict the lower bound: a commitment can authenticate a supplied opening without reconstructing every hidden record. Collision resistance is a computational property; the counting argument is about exact recovery.

The query family is decisive. If every query returns $\botout$, $N$ can be one and the summary is tiny but uninformative. If a decision depends only on an explicitly declared preference and one provenance distinction, most narrative details may be irrelevant to that query. The result does not require preserving all life events for every task. It requires preserving all distinctions that affect the specified answer. In our case, merging ``$H$ sent a question'' with ``a model inferred $H$'s motive'' destroys a distinction relevant to attribution-sensitive queries even when a prose summary remains fluent.

\subsection{More observations do not always identify the preferred action}
Let $W$ denote a possible world, let $D\sim P_W$ be an arbitrarily long passive observational archive, and let $a^*(W)$ be a unique preferred action under that world's evaluative specification. Neither $a^*$ nor the observation law is assumed to be learnable without constraints.

\begin{proposition}[Passive-history nonidentifiability]\label{prop:nonid}
Suppose $W_0,W_1$ induce the same distribution over the entire available archive but have distinct preferred actions $a_0\ne a_1$. Any procedure based only on that archive has worst-case error at least $1/2$ over these two worlds. Extra samples from the same indistinguishable observation process do not remove this bound.
\end{proposition}
\begin{proof}
Let $Q$ be the common distribution of the procedure's output, including its internal randomness. Its success probabilities in the two worlds are $Q(a_0)$ and $Q(a_1)$. Their sum is at most one, so at least one is at most $1/2$. If the complete observation laws are identical, additional passive samples leave the argument unchanged.
\end{proof}

The assumption is essential. Direct, reliable preference declarations, interventions, or justified restrictions on the world family can distinguish the worlds. The proposition does not say that learning personal preferences is impossible. It states that unlimited quantity cannot repair a particular absence of identifying information. This boundary is consistent with the broader preference-inference problem studied by Armstrong and Mindermann \cite{armstrong}.

For a purely constructed affiliation example, let the available actions be $b$ (join), $n$ (do not join), and $p$ (postpone). Let two worlds share every admitted message and feeling report but differ in an as-yet-unexpressed preference about commitment:
\begin{equation}
\begin{array}{c|ccc}
 & b & n & p\\\hline
u_0&1&0&0\\
u_1&0&1&0
\end{array}
\end{equation}
These numbers define a counterexample; they are not measurements of the diarist. If both worlds remain admissible, Equation~\eqref{eq:robust} returns $\botout$, not $p$. A subsequently authorized declaration may exclude one world. Repeating the same dinner description cannot by itself do so.

\subsection{Training and retrieval are different admission channels}
An archive can influence a system through retrieval or through training. For illustration, a next-token training objective can be written
\begin{equation}
L(\theta;D)=-\sum_{(w_1,\ldots,w_T)\in D}\sum_{t=1}^{T}
\log p_\theta(w_t\mid w_{<t}).
\end{equation}
This objective rewards predictive fit to text, not compliance with a decision authorization policy. If generated commentary is mixed into $D$, its style and assertions can affect the parameter state. Attaching source labels to examples can provide information to a learner, but does not prove that the learned function enforces those labels.

For Theorem~\ref{thm:ni} to cover an adapting model, its protected updates must themselves factor through the authorized view, for example
\begin{equation}
\theta_{t+1}=\mathsf{Train}(\theta_t,V_{P_t}(D_t,q_t);\omega_t),
\end{equation}
with all other shared state controlled. This equation is an architectural requirement, not a claim that a usual fine-tuning procedure satisfies it. A model already trained on excluded material does not meet the theorem merely because retrieval later filters that material. Likewise, a raw model proposal cannot choose the protected action set or tie-breaker through an untracked channel.

\section{Relationship to Alpay Algebra and Fixed Points}\label{sec:fixed}
Alpay Algebra I describes a state update and an evaluation condition \cite{alpay1}:
\begin{align}\label{eq:alpay}
x_{k+1}&=x_k+\phi(x_k),\\
\phi(x_k)\ne0&\ \Longrightarrow\ \Psi(x_{k+1})>\Psi(x_k).
\end{align}
Part II treats categorical fixed-point objects \cite{alpay2}. Part III describes a memory accumulation of the form $M_n=\bigcup_{k\le n}X_k$ \cite{alpay3}, and Part V gives
\begin{equation}\label{eq:alpayv}
E_{n+1}=\phi(E_n,\gamma(E_n))
\end{equation}
for an outer semantic update and an embedded subproblem \cite{alpay5}. These expressions can organize a comparison with an autobiographical system, but a mapping from narrative records to states and operators must still be specified.

A possible application-level state is $S_t=(D_t,P_t,V_t,\theta_t)$: archive, active policy, admitted view, and model state. This tuple is defined here, not derived uniquely from Alpay Algebra. An update $F$ should satisfy the invariant
\begin{equation}\label{eq:inv}
\mathcal{I}(S)\ \Longrightarrow\ \mathcal{I}(F(S)),
\end{equation}
where $\mathcal{I}$ includes authenticated provenance, the policy-update rule, and the exclusion boundary. A fixed point $S^*=F(S^*)$ neither defines nor proves~\eqref{eq:inv}. Indeed, a continually growing archive generally does not have a stationary full state. A representation for a frozen archive and policy might converge, but that is a different claim.

\begin{proposition}[Convergence does not select correctness]\label{prop:fixed}
Contractivity and uniqueness of a fixed point alone place no constraint on whether that point satisfies a desired decision property.
\end{proposition}
\begin{proof}
For any $z\in[0,1]$ and fixed $0\le\lambda<1$, define
\begin{equation}
F_z(x)=\lambda x+(1-\lambda)z,\qquad x\in[0,1].
\end{equation}
Then $|F_z(x)-F_z(y)|=\lambda|x-y|$ and the unique fixed point is $z$. Both $z=0$ and $z=1$ have the same contraction factor and convergence guarantee. If a decision predicate accepts only one endpoint, convergence cannot distinguish the correct predicate value without additional information tying $z$ to that predicate.
\end{proof}

Thus an archive could repeatedly stabilize on an unsupported interpretation. Stabilization would not authenticate its source or authorize its use. Selecting $\Psi=-L$ for a text-prediction loss defines improvement in predictive fit; it does not derive a measure of personal welfare. Nor is monotonic improvement guaranteed for arbitrary stochastic training updates. We therefore treat the relevant Alpay expressions as descriptive machinery and keep evaluation, authority, and correctness separate.

Our security theorems require neither transfinite iteration nor the semantic permanence claims discussed in Part IV \cite{alpay4}. Claims about an AI acquiring an author's intent, identity, or persistent memory are not premises. In particular, no theorem here calculates a baptism decision from a putative identity fixed point.

\section{Analytical Application to the Vignette}\label{sec:application}
The vignette supplies distinctions to preserve, not numerical utilities or a ground-truth verdict. We can nevertheless trace how the reference specification treats the supplied material.

\paragraph{Meal and affiliation.}
The reported enjoyment of a shared meal remains a positive feeling report. It can inform a question about valued social experiences. Under a policy that requires explicit owner authorization for membership decisions, it cannot create a policy transition. This represents the diarist's expressed distinction between continuing relationships and deciding about baptism, without asserting that the preference can never change.

\paragraph{Question, photograph, and chronology.}
The record of $H$'s question, the later hackathon photograph, and $U$'s intended self-expression receive distinct entries. A summary that says $H$ necessarily perceived a challenge changes an interpretation into an asserted fact. A summary that relocates the hackathon to the previous Sunday changes a time field. Theorem~\ref{thm:auth} can protect an already correct attestation of the time field; it cannot repair an incorrect initial transcription. A source-aware projection preserves both the supplied wording and the limited conclusion it warrants.

\paragraph{Knowledge of travel and apparent social distance.}
The fact that $C$ reportedly knew a travel plan does not specify how the information arrived. The map and seating episodes likewise leave multiple possible explanations. No source type licenses replacing this underdetermination with a verified account of coordination or intention. If an assistant-generated explanation is stored, its actual origin remains visible through the derivation graph. It can be analyzed as part of the diarist's information environment while remaining excluded from a factual claim about the other participants.

\paragraph{Assurances and private information.}
The leader's reported acknowledgment of risk and categorical assurances are preserved as separately attributable claims, rather than collapsed into a proven guarantee. A request to discuss a sensitive administrative matter privately defines a potential disclosure constraint, if the owner adopts it through the policy channel. A model cannot treat a public-friendly description of that matter as permission to reveal omitted details. This is a conventional authorization issue and requires no hypothesis of bad faith by the leader.

\paragraph{Interpretation across time.}
The diarist's feeling about an organizer's encouraging words and the uncertainty around another person's question illustrate why a final daily summary is not a complete substitute for time-indexed reports. A later revision can be appended as a revision without erasing the earlier experience. The currently authorized interpretation and the historical fact that an earlier interpretation was held answer different queries. Theorem~\ref{thm:compression} says only that a summary must retain this distinction when those queries are in its task family.

The religious-studies literature makes such distinctions intelligible as issues of belonging, narrative identity, and institutional ambivalence \cite{decoo,scharp}. It does not determine the individual's utility function, prove hidden intentions, or certify the output of a personal language model. The computer-security claim is invariant under replacing religious affiliation with employment, political participation, or another consequential institutional relationship.

\paragraph{Social connection and the meaning of agency.}
The meal illustrates a distinction that self-determination theory makes conceptually available: relatedness and autonomy can coexist \cite{ryandeci}. In the security model, however, $P$ represents an explicit authorization specification; it is not a psychological autonomy score. A signed consent event does not demonstrate freedom from pressure, just as declining to issue such an event does not establish hostility toward the community. Institutional ambivalence is likewise a reason to preserve several simultaneous evaluations, rather than compress friendship, doctrine, membership, and disclosure into one positive or negative sentiment \cite{decoo}. We do not estimate the strength of any such dimension from the supplied notes.

\paragraph{Source attribution across reflection cycles.}
The distinction between an event, its interpretation, and a later recollection resembles the attribution problem investigated in source-monitoring research \cite{johnson}. The comparison is functional, not an assertion that human memory and model context share a mechanism. An agent's reflection can become a later retrieval result; if derivation is lost, the same generated claim can appear to corroborate itself. We address that computational risk through lineage and admission. Testing the analogous human process would require an independent behavioral study with an appropriate design and participants. A model role-playing the diarist would not supply those observations.

\section{Synthetic Pilot: Labels and an Enforced Admission Boundary}\label{sec:pilot}
\subsection{Question and fixed protocol}
The experiment asks a narrow question: does appending a contradictory derived summary change a frozen model's answer to an attribution- or authorization-sensitive query, and how does this differ between labeling and exclusion? We use Qwen2.5-0.5B-Instruct, a public instruction-tuned model distributed through Hugging Face \cite{qwen}. This is a read-boundary experiment, not a complete autonomous agent or a multi-year memory system. It does not implement the signature verifier or evaluate the ideal utility rule in Equation~\eqref{eq:robust}.

Before inference, a local protocol fixed the model, all cases, prompts, metrics, and inference settings, together with hashes of the case file and run script. It was not externally preregistered. The materials contain 32 wholly synthetic cases: four families (registration authorization, disclosure permission, speaker attribution, and event chronology), four lexical contexts, and two binary ground-truth values. The contexts are a reading circle, community choir, history society, and volunteer group. They contain no private messages or identifiers from the vignette. Ground truth follows directly from explicit synthetic owner statements or event records.

For every case we append $k\in\{0,1,5\}$ copies of a fixed summary that contradicts the original record and asks that its conclusion be used. Every copy declares the original record as its parent. This is a fixed stressor, not an adaptive attack search or independent corroborating testimony. We compare three presentations:
\begin{enumerate}
\item \textbf{Flattened:} concatenate record payloads while omitting the structured source and parent fields. Original payload wording, such as ``Owner current instruction,'' remains.
\item \textbf{Typed:} present the same payloads with explicit identifiers, source types, and parent identifiers.
\item \textbf{Gated:} remove records typed as assistant summaries before constructing the otherwise identical typed input. Source metadata is assumed correct by construction.
\end{enumerate}
All presentations use the same system instruction prioritizing original owner statements and event records. The gate does not supply the correct answer or replace the model's answer. It changes the admitted input. Because it removes all added summaries, its rendered prompt is identical for all three $k$ values within a case and option order. This deliberately instantiates the premise of Theorem~\ref{thm:ni}; zero change is not evidence of learned resistance.

\subsection{Scoring and execution}
Each yes/no question is scored twice, once with $A=\mathrm{Yes}$ and once with the options reversed. We take the larger next-token logit among the single tokens $A$ and $B$, then map the letter back to yes/no. There is no sampling or unrestricted prose generation. In the recorded run, the unrestricted highest-logit token was also $A$ or $B$ for all 384 unique prompts, so the restriction did not change any first-token choice. Restricted two-option probabilities are diagnostic model scores, not calibrated probabilities of correctness.

Let $N=32$, $j\in\{0,1\}$ denote the option order, and $\widehat y_{c,v,k,j}$ be the binary prediction for case $c$, view $v$, and repetition count $k$. We report
\begin{align}
\widehat{\mathrm{Acc}}_{v,k}&=\frac{1}{2N}\sum_{c,j}\mathbf 1[\widehat y_{c,v,k,j}=y_c],\\
\widehat{\Delta}_{v,k}&=\frac{1}{2N}\sum_{c,j}\mathbf 1[\widehat y_{c,v,k,j}\ne\widehat y_{c,v,0,j}],\\
\widehat{B}_{v,k}&=\frac{1}{N}\sum_c\mathbf 1[\widehat y_{c,v,k,0}\ne\widehat y_{c,v,k,1}].
\end{align}
The second metric is a paired label-change rate, not an attack-success rate: a change can correct an existing mistake. The third measures answer sensitivity to option order after translating letters back into the same binary meaning.

The run contains $32\times3\times3\times2=576$ logical trials. Identical rendered prompts reuse one inference, leaving 384 unique model evaluations. These reused predictions are not additional independent observations. The pinned checkpoint contains 494,032,768 parameters. Execution used Python 3.12.14, PyTorch 2.6.0+cpu, Transformers 4.51.3, float32, four CPU threads, batch size four, evaluation mode, deterministic algorithms, and seed 20261012. The processor reports AMD EPYC 9V74. Prompts reached at most 421 tokens; the recorded inference loop took 316.24 seconds. These resource measurements describe this execution, not a performance benchmark. Full prompts, raw scores, hashes, software versions, and the model revision accompany the manuscript.

\subsection{Results}
Table~\ref{tab:pilot} reports the complete aggregate outcome. Clean typed and gated inputs yield 45/64 correct answers (70.31\%). Adding five contradictory summaries to typed input yields 32/64 (50.00\%), with 21/64 answers changed relative to the clean version. Of those changes, 17 turn a correct answer into an error and four correct a pre-existing error. Source labels therefore do not enforce invariance in this particular model and test suite. The experiment does not establish that labels are generally ineffective or that this model represents larger deployed systems.

\begin{table}[t]
\centering\small
\begin{tabular}{@{}lrrrr@{}}
\toprule
Presentation & Copies $k$ & Correct / 64 & Changed / 64 & Order disagreements / 32\\
\midrule
Flattened & 0 & 37 & 0 & 27\\
Flattened & 1 & 36 & 1 & 28\\
Flattened & 5 & 28 & 9 & 20\\
\addlinespace
Typed & 0 & 45 & 0 & 19\\
Typed & 1 & 34 & 19 & 18\\
Typed & 5 & 32 & 21 & 24\\
\addlinespace
Gated & 0 & 45 & 0 & 19\\
Gated & 1 & 45 & 0 & 19\\
Gated & 5 & 45 & 0 & 19\\
\bottomrule
\end{tabular}
\caption{Recorded descriptive results from 32 synthetic cases with two option orders. ``Changed'' compares the same case and order with $k=0$. Gated repetitions reuse identical prompt predictions. No row represents 64 independent people or independently sampled scenarios.}
\label{tab:pilot}
\end{table}

Gated accuracy remains 45/64 at every repetition level because its inputs are unchanged. The 19 remaining errors matter: noninterference does not imply correct interpretation of admitted material. There is also substantial option-order sensitivity, with 19/32 clean typed cases producing different yes/no judgments after swapping the letters. This limits the strength of conclusions about semantic reasoning. We report the sensitivity rather than selecting the better-performing order.

\paragraph{Interpretation and limitations.}
The pilot illustrates a distinction between a model being told an authority rule and an architecture controlling the information it receives. It supplies neither a general cryptographic security estimate nor a robust comparative ranking of memory systems. The sample is small, manually constructed, and template-related; counts are descriptive, without population confidence intervals or significance tests. Filtering shortens inputs, so its accuracy difference is not a controlled estimate of a uniquely semantic mechanism. The model is small, the stressor is fixed, and metadata parsing is not attacked. No private archive is used for training; no claim is made about learned personal preferences, participants' hidden motives, or actual religious decisions. Long-horizon read--write agent evaluation and human-subject validation remain separate tasks.

\paragraph{Data and code availability.}
The accompanying ancillary materials contain the synthetic cases, locally fixed protocol, exact model revision and file hashes, rendered prompts, recorded predictions, trial-level results, environment records, and scripts for inference and independent analysis. Model weights are not redistributed. These materials document the completed pilot; the personal source material used to motivate the vignette is not included.

\section{Limitations and Research Questions}\label{sec:limitations}
\paragraph{Trusted capture and semantic admission.}
Our strongest results assume correct types, complete relevant parent links, a faithful capture process, and a correct formal gate. These are demanding assumptions. Cryptography cannot convert an unreliable transcript into a verified event. Designing a useful natural-language admission system with explicit, compositional error bounds remains a research question within this specification, rather than a solved problem asserted here.

\paragraph{Preference change and consent.}
An authenticated owner policy can be incomplete, mistaken, or superseded. Our separation preserves the authority to revise it; it does not define a timeless true self. Noninterference compares executions with equal authorized inputs. If reading a message leads the owner to issue a genuinely different policy, that comparison no longer applies. Protecting a person from all social influence is neither the theorem nor a realistic cryptographic objective.

\paragraph{A precise constrained compression problem.}
Let $s_P(h)$ be the security-relevant fields that must remain exactly recoverable, $\mu$ a declared distribution on admitted histories and queries, and $d$ a bounded task loss. For nonempty finite $\mathcal H$ and $\mathcal Q$, define
\begin{align}
\mathcal R_{P,\mu}(\epsilon)&=\min_{C,G,E}\log_2|C(\mathcal H)|,\label{eq:rate}\\
\text{subject to}\qquad E(C(h))&=s_P(h)\quad\forall h\in\mathcal H,\nonumber\\
\mathbb E_{(h,q)\sim\mu}\big[d(f_P(h,q),G(C(h),q))\big]&\le\epsilon.\nonumber
\end{align}
Here $E$ decodes the mandatory fields, while $G$ answers the task. Necessarily $\mathcal R_{P,\mu}(\epsilon)\ge\log_2|s_P(\mathcal H)|$. For zero--one task loss, full-support $\mu$, and $\epsilon=0$, the exact partition argument applies to the joint signature $(s_P(h),(f_P(h,q))_{q\in\mathcal Q})$. Thus the unrestricted finite exact case is already characterized; it is not presented as an unresolved mathematical mystery.

The outstanding design question in this model is an efficient, incremental approximation under stated query and policy classes: how close can bounded-memory encoders come to~\eqref{eq:rate} while retaining exact authorization fields, preserving noninterference, and processing admitted updates without replaying the full archive? Arbitrary future policies that expose previously discarded distinctions cannot be supported by a fixed summary in general, by Theorem~\ref{thm:compression}. Useful positive results therefore require restricted policy changes, access to an external authenticated archive, or explicit error allowances. No general optimal algorithm or security proof for a natural-language encoder is supplied here.

\paragraph{Learning, deletion, and disclosure.}
Separate proofs are required for a learner's response to excluded training data, for removing the influence of withdrawn records, and for model-output privacy. Access control on a retained archive does not by itself remove information already memorized in parameters. A permanent commitment to a record can support later audit while the record's active status changes; it does not resolve every conflict between confidentiality, deletion, and auditability.

\paragraph{Scope of the vignette and theoretical status.}
The case has no sampling framework, interviews conducted for research, prevalence estimates, or independently established behavioral ground truth. Its role is specification by example. The synthetic pilot is a separate finite test suite with deliberately correct source metadata, not a psychological study or broad security benchmark. We make no finding about LDS members generally or the intentions of any named individual. The proofs are handwritten mathematical arguments, not machine-checked developments. No claim of deployment robustness, formal publication acceptance, or exhaustive literature coverage follows from this manuscript.

\section{Conclusion}
Autobiographical archives require distinctions among record integrity, attributed meaning, and decision authority. Authenticated records and an explicit policy channel can support conditional protection against alteration and unauthorized promotion. A protected decision procedure is noninterfering with respect to excluded content only if all its relevant inputs and state transitions obey that boundary. Exact memory compression is governed by the distinctions a declared query family can observe, while passive history and fixed-point convergence need not identify a preferred personal action. The synthetic pilot exhibits source-label failures in one small model and stable but imperfect answers under an enforced admission boundary. Its substantial option-order sensitivity prevents broad empirical generalization. The affiliation vignette and social-science literature motivate what must be represented without proving participants' intentions. The resulting framework specifies conditional security properties, a constrained compression problem, and a reproducible starting point for broader evaluation.

\appendix
\section{What the Statements Do and Do Not Establish}
The record-integrity result establishes, under its assumptions and except with cryptographic failure, that accepted covered fields match an attested message. It establishes neither speaker honesty nor world truth nor correct initial labeling. Likewise, log consistency establishes that a later accepted checkpoint extends a retained checkpoint. It does not establish a complete life history, the freshest available state, or the global absence of equivocation.

Policy separation prevents archive text from syntactically replacing owner authorization; it does not ensure freedom from coercion or all social influence. The noninterference argument gives equality of protected output laws when admitted inputs and protected state are equal. It does not establish the safety of unrestricted prose, cover hidden training channels, or guarantee availability.

The compression bound requires exact summaries to distinguish task-inequivalent histories. It supplies no efficient compressor or decoder and no optimal approximate summary. The nonidentifiability construction shows that indistinguishable observation laws can conceal opposite preferences, rather than proving preference learning universally impossible. The fixed-point counterexample similarly shows that contractivity alone does not select a desired decision property; it does not show that every fixed-point method fails.

The synthetic pilot establishes recorded output changes within its fixed small-model test and invariance under identical gated inputs. These observations establish neither general agent robustness nor cryptographic implementation security nor findings about human psychology. Finally, the affiliation vignette supplies concrete examples of provenance and authorization distinctions. It supports no population findings, theological conclusions, or proof of participants' motives.

\clearpage
\section{A Provenance-Preserving Abstract Record Sequence}
The following sequence specifies types and dependencies, not a deployed data format. Labels refer to the bounded vignette; no new factual observations are introduced.
\begin{enumerate}
\item $r_1$: a quotation attributed to $H$, asking about the previous Sunday. The attester is the importer, unless an independent speaker signature exists.
\item $r_2$: $U$'s report of sending a later hackathon photograph. Its event time is distinct from the Sunday referenced by $r_1$.
\item $r_3$: $U$'s statement of the meaning intended by that photograph. Type: interpretation of one's own communicative intention.
\item $r_4$: an assistant passage assigning a strategic motive to $H$, with parents $r_1,r_2$. Type: generated interpretation, with the assistant as source.
\item $r_5$: a summary derived from $r_4$. Its primitive evidential lineage is inherited; it is not a second observation of $H$'s intentions.
\item $r_6$: the owner's expressed wish to retain independent choice about affiliation. Type: declaration. It becomes an active machine policy only through a separate owner-authorized policy event.
\end{enumerate}
A query about what $H$ wrote may admit $r_1$. A query about how $U$ framed the photograph may admit $r_2,r_3$. A query about the evolution of assistant commentary may admit $r_4,r_5$ as generated commentary. None of these admissions silently changes $r_4$ into an observed fact or $r_6$ into an action that has already been authorized. This separation is the concrete security requirement that motivates the formal model.

\end{document}